\documentclass{IEEEtran}

\usepackage{amsmath}
\usepackage{amsfonts}
\usepackage{amssymb}
\usepackage{bm}

\usepackage{amsthm}
\newtheorem{lem}{Lemma}[section]
\newtheorem{theorem}{Theorem}[section]
\newtheorem{cor}{Corollary}[section]
\newtheorem{defn}{Definition}[section]
\newcounter{MYtempeqncnt}
\newcounter{magicEquations}

\usepackage{graphicx}

\title{Cram\'{e}r-Rao Bound for Blind Channel Estimators in Redundant Block Transmission Systems}
\author{Yen-Huan Li, {\it Student Member, IEEE}, Borching Su, {\it Member, IEEE}, and Ping-Cheng
Yeh, {\it Member, IEEE}}

\begin{document}
\maketitle

\begin{abstract}
In this paper, we derive the Cram\'{e}r-Rao bound (CRB) for blind
channel estimation in redundant block transmission systems, a lower
bound for the mean squared error of any blind channel estimators.
The derived CRB is valid for any full-rank linear redundant
precoder, including both zero-padded (ZP) and cyclic-prefixed (CP)
precoders. A simple form of CRBs for multiple complex parameters is
also derived and presented which facilitates the CRB derivation of
the problem of interest. A comparison is made between the derived
CRBs and performances of existing subspace-based blind channel
estimators for both CP and ZP systems. Numerical results show that
there is still some room for performance improvement of blind
channel estimators.
\end{abstract}


\section{Introduction}
Block transmission systems, especially orthogonal frequency-division
multiplexing (OFDM) systems, have become one of the most popular
solutions to meet the high data rate requirements of modern
communication standards. For example, wireless local-area networks
(WLAN) standards such as IEEE 802.11n \cite{802.11n}, and beyond
third-generation (B3G) cellular communication standards such as IEEE
802.16e \cite{IEEE802.16e} or 3GPP-LTE \cite{3GPPLTE}, all apply
block transmission schemes as the basic physical-layer transmission
scheme. Channel estimation is crucial for equalization in these
systems. Most of the modern standards using block transmission
systems insert some already-known symbols, called pilot symbols, in
the transmitted signals, and the receiver estimates the channel
response from the pilot symbols \cite{Tong2004}. Such an approach is
called pilot-assisted channel estimation method. However, pilot
assisted methods suffer from loss of bandwidth efficiency since the
inserted pilot symbols do not carry any information.

Blind channel estimation in block transmission systems, on the other
hand, aims to avoid the redundancy introduced by pilot symbols. The
goal of blind channel estimation is to estimate the channel response
directly from \emph{unknown symbols}, which can be data symbols
\footnote{Theoretically, we only need one pilot symbol to eliminate
the scalar ambiguity \cite{Su2007}.}. Current blind channel
estimation algorithms can be roughly divided into two main
categories. The first exploits the fact that all of the transmitted
symbols come from a set of finite number of points in the signal
space, i.e., the modulation constellation \cite{Zhou2001}. The
second assumes no \textit{a priori} information about the modulation
constellation at the receiver side and consists mostly of
subspace-based methods \cite{Tong1998}. This kind of approach can be
applied in a wider range of situations, but may have a slightly
worse performance than finite-alphabet methods. This paper addresses
only the performance bound of blind channel estimators assuming no
{\it a priori} information about the transmitted signals.

Cram\'{e}r-Rao bound (CRB) is an important performance bound which
gives a lower bound for the mean squared error (MSE) of channel
estimators in a communication system \cite{Tong2004}. A CRB for
blind finite impulse response (FIR) multichannel estimation has been
derived in \cite{Carvalho1997}. It is, however, not applicable for
redundant block transmission systems. There are two main types of
redundant block transmission systems, namely, systems with null
guard intervals, also known as zero padded (ZP) systems, and systems
with cyclic prefixes (CP). The CRB for blind channel estimators in
ZP block transmission systems has been derived in
\cite{Barbarossa2002, Su2008a}. But the CRB for blind channel
estimators in CP block transmission systems, to our best knowledge,
has not yet been studied in the literature. In this paper, a general
CRB for block transmission systems is derived, which is valid for
any linear redundant precoders, including those in ZP
\cite{Scaglione1999a} and CP systems. We then compare the
performances of existing blind channel estimators \cite{Su2007,
Su2007a, Muquet2002, Scaglione1999b} with the derived CRB.

An additional contribution of this paper is a simplification of CRB
formulas for complex parameters. To calculate the CRB for blind
channel estimation in wireless communication systems, in which
numerical values are usually modeled by complex numbers, we extend
the existing results of CRBs for unconstrained and constrained
parameters, originally for \emph{real} parameters
\cite{VanTrees2001,Gorman1990,Stoica1998}, to the case of multiple
\emph{complex} parameters. We devote one section to do this work
before starting the derivation of the CRB for the problem of
interest.

There has been some existing literature on this topic
\cite{Bos1994,Jagannatham2004,Smith2005,Ollila2008}. In
\cite{Bos1994,Jagannatham2004,Ollila2008}, the derived CRBs for
unconstrained and constrained complex parameters, unlike the
corresponding CRBs for real parameters, are variance bounds of any
unbiased estimators for $[\begin{array}{cc}\bm{\theta}^T &
\bm{\theta}^H \end{array}]^T \in \mathbb{C}^{2n}$, a vector of
double size of the unknown parameter vector
$\bm{\theta}\in\mathbb{C}^n$. In \cite{Smith2005}, a CRB that
represents variance bound of unbiased estimators for
$\bm{\theta}\in\mathbb{C}^n$ is presented for the first time.
However, its result is in a complicated form compared with the well
known CRB for real parameters \cite{Stoica1998}. In addition, the
result in \cite{Smith2005} does not consider constraints on unknown
parameters and cannot be directly applied in the problem considered
in this paper.

Instead of a complicated form, we seek here to derive {\it simple}
CRBs for unconstrained and constrained complex parameters in a form
similar to those for real parameters. Specifically, the form of the
derived CRB for unconstrained complex parameters is exactly
identical to the CRB for unconstrained real parameters; and if the
constraint function is holomorphic, the CRB for constrained
complex parameters can also have the same form as that of the CRB
for constrained real parameters \cite{Stoica1998}. Using this
result, the derivation of CRB for the blind channel estimation
problem becomes simple and compact compared to that in
\cite{Barbarossa2002}.

The rest of the paper is organized as follows. In Section
\ref{sec_complex_crb}, we derive \emph{simple} CRBs for
unconstrained and constrained complex parameters. In Section
\ref{sec_sys_model}, the system model of a redundant block
transmission system is described and the blind channel estimation
problem is formulated. We then derive the CRB for blind channel
estimators in Section \ref{sec_crb_blind}, using results given in
Section \ref{sec_complex_crb}. In Section \ref{sec_num_results},
numerical results are conducted to compare the derived CRB with
performances of existing blind channel estimators. Conclusions and
future work are presented in Section \ref{sec_conclusions}.

\subsection*{Notations}
Bold-faced lower case letters represent column vectors, and
bold-faced upper case letters are matrices. Superscripts such as in
$\bm{v}^*$, $\bm{v}^T$, $\bm{v}^H$, $\bm{M}^{-1}$, and
$\bm{M}^{\dagger}$ denote the conjugate, transpose, conjugate
transpose (Hermitian), inverse, and Moore-Penrose generalized
inverse of the corresponding vector or matrix. The vector
$\mathsf{vec}( \bm{M} )$ denotes the column vector formed by
stacking all columns of $\bm{M}$. The vector $\mathsf{E}\left[
\bm{v} \right]$ denotes the expectation of $\bm{v}$, and the matrix
$\mathsf{E}\left[ \bm{M} \right]$ denotes the expectation of
$\bm{M}$. The matrix $\mathsf{cov}(\bm{u}, \bm{v})$ denotes the
cross-covariance matrix of random vectors $\bm{u}$ and $\bm{v}$ and
is defined as $\mathsf{cov}(\bm{u}, \bm{v}) \triangleq \mathsf{E}
\left[ ( \bm{u} - \mathsf{E}( \bm{u} ) ) ( \bm{v} - \mathsf{E}(
\bm{v} ) )^H \right]$. The Kronecker product of $\bm{A}$ and
$\bm{B}$ is denoted by $\bm{A} \otimes \bm{B}$. The notation $\bm{A}
\geq \bm{B}$ means that $\bm{A} - \bm{B}$ is a nonnegative-definite
matrix. The trace of a square matrix $\bm{A}$ is denoted by
$\mathsf{tr}(\bm{A})$. Matrices $\bm{I}_n$ and $\bm{0}_{m\times n}$
denote the $n\times n$ identity matrix and the $m\times n$ zero
matrix, respectively. Notations $[\bm{A}]_{i, j}$ and $[\bm{v}]_i$
refer to the $(i,j)$th entry of matrix $\bm{A}$ and the $i$th
element of vector $\bm{v}$, respectively. The notation
$[\bm{v}]_{a:b}$ denotes the column vector whose elements contain
the $a$th through $b$th elements of vector $\bm{v}$.

\section{Simple Forms of CRBs for Complex Parameters} \label{sec_complex_crb}
In this section, we present an extension of CRBs originally for
unconstrained and constrained \emph{real} parameters to the case of
\emph{complex} parameters. The CRB for constrained complex
parameters derived in this section, in particular, will be directly
used in Section \ref{sec_crb_blind} for CRB derivation for blind
channel estimators in redundant block transmission systems.

\subsection{CRB for Unconstrained Complex Parameters}
Suppose that $\bm{\theta} \in \mathbb{C}^n$ is an unknown complex
parameter vector, and $\bm{t}(\bm{y})$ an unbiased estimator of
$\bm{\theta}$ based on a complex observation vector $\bm{y} \in
\mathbb{C}^p$ characterized by a probability density function (pdf)
$p( \bm{y}; \bm{\theta} )$. We assume the regularity condition is
met:
\begin{equation}
\mathsf{E}\left[ \frac{\partial \ln p(\bm{y};\bm{\theta})}{\partial
\boldsymbol{\theta}^T} \right] = \mathsf{E}\left[ \frac{\partial \ln
p(\bm{y};\bm{\theta})}{\partial \boldsymbol{\theta}^H} \right] =
\boldsymbol{0}.
\end{equation}
Here the differentiation is defined according to Wirtinger's
calculus \cite{Fritzsche2002} that
\begin{equation}
\frac{\partial f}{\partial z} \triangleq \frac{1}{2} \left(
\frac{\partial f}{\partial \alpha} - j\frac{\partial f}{\partial
\beta} \right)
\end{equation}
for all $f: \mathbb{C}\to\mathbb{C}$ and $z = \alpha + j\beta \in
\mathbb{C}$, $\alpha, \beta \in \mathbb{R}$. We assume
$p(\bm{y};\bm{\theta)}$ to be real differentiable so that the
derivative exists. We first define the complex Fisher information
matrix (FIM) and then give the CRB expression in terms of the FIM as
a theorem.

\begin{defn}[Complex Fisher information matrix] Suppose $\boldsymbol{\theta} \in \mathbb{C}^n$
is an unknown complex parameter vector and $\boldsymbol{y}$ is a
complex observation vector characterized by a pdf
$p(\boldsymbol{y};\boldsymbol{\theta})$. Then the \emph{complex
Fisher information matrix} (FIM) for $\boldsymbol{y}$ is defined as
\begin{equation}
\bm{J} \triangleq \mathsf{E} \left[ \left( \frac{\partial \ln
p}{\partial \boldsymbol{\theta}^*}\right) \left( \frac{\partial \ln
p}{\partial \boldsymbol{\theta}^*}\right)^H \right].
\label{eq_def_fim}
\end{equation}
\end{defn}
Note that the FIM defined above has a size of $n\times n$ rather
than $2n\times 2n$ as in many previous results
\cite{Bos1994,Jagannatham2004,Ollila2008}. We then present the CRB
expression of unconstrained complex parameters in terms of this FIM
in the following theorem.

\begin{theorem} \label{thm_complex_crlb}
Suppose $\boldsymbol{t}(\boldsymbol{y}) \in \mathbb{C}^n$ is an
unbiased estimator of an unknown complex parameter vector
$\boldsymbol{\theta} \in \mathbb{C}^n$ based on a complex
observation vector $\boldsymbol{y}$ (i.e.,
$\mathsf{E}\left[\boldsymbol{t}(\boldsymbol{y})\right] =
\boldsymbol{\theta}$ ) characterized by a pdf
$p(\boldsymbol{y};\boldsymbol{\theta})$. Then
\begin{equation}
\mathsf{cov}\left( \boldsymbol{t}, \bm{t} \right) \geq
\bm{J}^{\dagger}. \label{new_crlb}
\end{equation}
 The equality holds if and only if
\begin{equation}
\boldsymbol{t} - \boldsymbol{\theta} = \boldsymbol{J}^{\dagger}
\frac{\partial \ln p}{\partial \boldsymbol{\theta}^*} \mbox{ with
probability }1. \notag
\end{equation}

\end{theorem}
\begin{proof}
See Appendix A.
\end{proof}

\subsection{CRB for Constrained Complex Parameters}
In this subsection we derive the CRB for constrained complex
parameters with \emph{holomorphic constraint functions}. Although
the result is not so general as those in \cite{Jagannatham2004,
Ollila2008}, it is more compact and easy to manipulate. In fact, the
form of the derived constrained CRB for complex parameters is the
same as that for real parameters \cite{Stoica1998}.

To begin the derivation of complex constrained CRB, we need some
basic concepts from complex analysis, including the definition of a
holomorphic map.

\begin{defn}[Holomorphic map \cite{Fritzsche2002}]
Let $B$ be an open subset of $\mathbb{C}^n$. A map $\bm{f}: B\to
\mathbb{C}^m$ is said to be holomorphic if
\begin{equation}
\frac{\partial \bm{f}}{\partial \bm{z}^H} = \bm{0},
\end{equation}
where $\bm{z}\in\mathbb{C}^n$ is a complex vector.
\end{defn}

As long as the constraints can be expressed as a holomorphic map of
complex parameters, the CRB for constrained complex parameters
exists, as presented in the following theorem.
\begin{figure*}[!t]
\begin{center}
\includegraphics[width=5.4in]{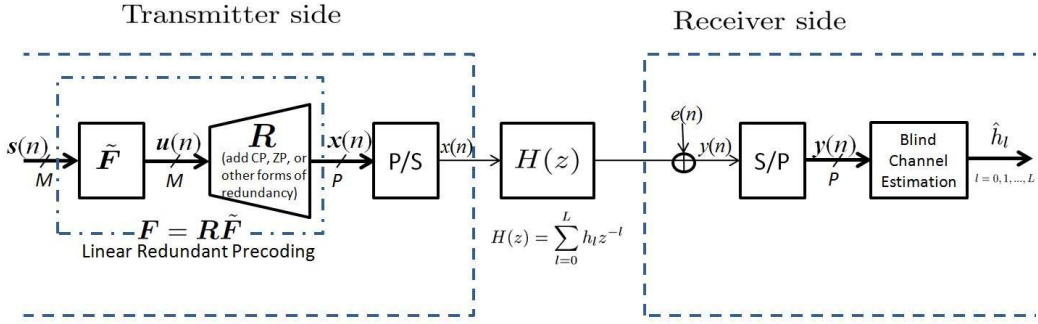}\label{fig:block_diagram}
\centering\caption{Block diagram of blind channel estimation in a
redundant block transmission system} \end{center}
\end{figure*}
\begin{theorem} \label{them_complex_constrained_crb}
Let $\bm{t}( \bm{y} )$ be an unbiased estimator of an unknown
parameter $\bm{\theta} \in \mathbb{C}^n$ based on observation
$\bm{y} \in \mathbb{C}^p$ characterized by its pdf $p( \bm{y};
\bm{\theta} )$. Furthermore, we require the parameter $\bm{\theta}$
to satisfy a holomorphic constraint function $\bm{f}: \mathbb{C}^n
\to \mathbb{C}^m$, $m\leq n$,
\begin{equation}
\bm{f}( \bm{\theta} ) = \bm{0}.
\end{equation}
Assume that $\partial \bm{f} / \partial \bm{\theta}^T$ has full
rank. Let $\bm{U}$ be a matrix with $(n-m)$ orthonormal columns that
satisfies
\begin{equation}
\frac{\partial \bm{f}}{ \partial \bm{\theta}^T } \bm{U} = \bm{0}.
\label{eq:orthonormalU}
\end{equation}
Then
\begin{equation}
\mathsf{cov}\left( \bm{t}, \bm{t} \right) \geq \bm{U} \left(
\bm{U}^H \bm{J} \bm{U} \right)^\dagger \bm{U}^H,
\end{equation}
where $\bm{J}$ is the FIM defined as in (\ref{eq_def_fim}). The
equality holds if and only if
\begin{equation}
\bm{t} - \bm{\theta} = \bm{U} \left( \bm{U}^H \bm{J} \bm{U}
\right)^\dagger \bm{U}^H \frac{\partial \ln p}{\partial
\bm{\theta}^*}\quad \mbox{with probability }1.
\end{equation}
\end{theorem}
\begin{proof}
See Appendix B.
\end{proof}

With Theorem \ref{them_complex_constrained_crb}, we are ready to
derive the CRB for blind channel estimators in redundant block
transmission systems, as shown in the following sections.

\section{System Model and Problem Formulation} \label{sec_sys_model}
In this section we formulate the blind channel estimation problem in
a redundant block transmission system using an equivalent
discrete-time baseband model.

The block diagram of a block transmission system is shown in Figure
1. Let the $n$th source block be expressed as an $M\times 1$ column
vector
\begin{equation} \bm{s}(n) = [\begin{array}{cccc} s_0(n) & s_1(n) &
\cdots & s_{M-1}(n)
\end{array}]^T,
\end{equation}
where each entry in the vector is a modulation symbol expressed as a
complex value. The vector $\bm{s}(n)$ is precoded by a full-rank
$M$-by-$M$ matrix $\tilde{\bm{F}}$ to obtain
\begin{equation}
\bm{u}(n) = \tilde{\bm{F}} \bm{s}(n).
\end{equation}
In an OFDM system, $\tilde{\bm{F}}$ equals to $\bm{W}^H$, the
normalized inverse discrete Fourier transform (IDFT) matrix; in a
single carrier (SC) system, $\tilde{\bm{F}}$ equals to the identity
matrix $\bm{I}_M$. We assume the general case that $\tilde{\bm{F}}$
is nonsingular as in \cite{Barbarossa2002}. For convenience we
define
\begin{equation}
P \triangleq M + L.
\end{equation}
The precoded $M$-vector $\bm{u}(n)$ is then added redundancy to
obtain a $P$-vector $$\bm{x}(n) = \bm{R}\bm{u}(n)$$ where $\bm{R}$
is a $P\times M$ full rank matrix representing the type of
redundancy added. In a cyclic prefix (CP) system, we have
\begin{equation}
\bm{R} = \left[ \begin{array}{cc}
\bm{0}_{L\times (M-L)} & \bm{I}_L \\
\bm{I}_{M-L} & \bm{0}_L \\
\bm{0}_{L \times (M-L)} & \bm{I}_L
\end{array} \right],
\label{eq:Rcp}
\end{equation}
while in a ZP system, we have
\begin{equation}
\bm{R} = \left[ \begin{array}{c} \bm{I}_M \\ \bm{0}_{L\times M}
\end{array} \right].\label{eq:Rzp}
\end{equation}
In general, we have a unified expression for $\bm{x}(n)$,
\begin{equation}
\bm{x}(n) = \bm{R}\tilde{\bm{F}}\bm{s}(n)=\bm{F}
\bm{s}(n)\label{eq:xFs}
\end{equation}
where \begin{equation} \bm{F} \triangleq \bm{R}\tilde{\bm{F}}
\label{eq:F=RF_tilde}
\end{equation}
is a full-rank $P\times M$ matrix. Note that ZP-OFDM, SC-ZP,
CP-OFDM, and SC-CP systems are all special cases of the redundant
block transmission system considered here.

After parallel-to-serial conversion, $\bm{x}(n)$ is sent to the
channel, modeled as a finite impulse response (FIR) filter
characterized by its $z$-transform
\begin{equation}
H(z) = \sum_{l=0}^L h_l z^{-l}.
\end{equation}
The output of the FIR channel is the convolution of the input and
the channel impulse response. Suppose the transmitter sends $N$
consecutive blocks defined as
\begin{equation}
\bm{s}_N \triangleq [\begin{array}{cccc}
\bm{s}(0)^T & \bm{s}(1)^T & \ldots & \bm{s}(N-1)^T \end{array}]^T.
\end{equation}
The $n$th received block, $\bm{y}(n)$, $n=0, 1, 2, ...,
N-1$, can be expressed as
\begin{equation}
\bm{y}(n) =
\bm{H}_1\left[\begin{array}{c}\left[{\bm{x}}(n-1)\right]_{M+1:P}\\\bm{x}(n)\\\end{array}\right]
+ \bm{e}(n),\label{eq:yGHxen}
\end{equation}
where $\bm{H}_1$ is a $P$-by-$(P+L)$ Toeplitz matrix \cite{Horn1985}
of the form
\begin{equation}
\bm{H}_1 \triangleq \left[
\begin{array}{ccccccc}
h_L& \cdots& h_1& h_0 & 0 & \cdots & 0 \\
0 & h_L & \cdots & h_1 & h_0 & \ddots & \vdots \\
\vdots& \ddots& \ddots&  & \ddots & \ddots & 0 \\
0& \cdots& 0& h_L & \ddots & h_1 & h_0 \\
\end{array}
\right], \label{fatToeplitz}
\end{equation}
and $\bm{e}(n)$ is the additive white Gaussian noise (AWGN) at the
receiver side with zero mean and covariance matrix
\begin{equation}
\mathsf{E}[\bm{e}(n)\bm{e}^H(n)] = \sigma^2 \bm{I}_P.
\end{equation}
Notice that, for $n=0$, the values of the first $L$ entries of
$\bm{y}(0)$ depend on undetermined vector $\bm{x}(-1)$. So for the
channel estimation problem, we drop these $L$ samples from the
observation and use only the last $M$ samples of the $0$th block
$\bm{y}(0)$, which can be expressed as
\begin{equation}
[\bm{y}(0)]_{L+1:P} = \bm{H}_2\bm{x}(0)
 + \bm{e}(0)\label{eq:yGHxe0}\end{equation} where $\bm{H}_2$ is an $M\times P$
Toeplitz matrix of the same form as $\bm{H}_1$ in
(\ref{fatToeplitz}). Collect the $(NP-L)$ samples observed by the
receiver and define the observation vector as
\begin{equation}
\bm{y}_N \triangleq [\begin{array}{cccc} [\bm{y}(0)]_{L+1:P}^T &
\bm{y}(1)^T & \ldots & \bm{y}(N-1)^T \end{array}]^T. \label{main_eq}
\end{equation}
Then it can be shown, from (\ref{eq:xFs}), (\ref{eq:yGHxen}), and
(\ref{eq:yGHxe0}), that
\begin{equation}
\bm{y}_N = \bm{G} \bm{H} ( \bm{I}_N \otimes \bm{F} ) \bm{s}_N +
\bm{e}_N \label{main_eq1}
\end{equation} where \begin{equation} \bm{G} \triangleq
\left[\begin{array}{ccc} \bm{0}_{(NP-L)\times L} &
\bm{I}_{(NP-L)\times (NP-L)} & \bm{0}_{(NP-L\times L)}
\end{array} \right],
\end{equation}
$\bm{H}$ is a $(NP+L)$-by-$NP$ Toeplitz matrix of the form
\begin{equation}
\bm{H} \triangleq \left[
\begin{array}{cccc}
h_0 & 0 & \cdots & 0 \\
h_1 & \ddots & \ddots & \vdots \\
\vdots & \ddots & \ddots & 0 \\
h_L & \ddots & \ddots & h_0 \\
0 & \ddots & \ddots & h_1 \\
\vdots & \ddots & \ddots & \vdots \\
0 & \ldots & 0 & h_L
\end{array}
\right], \label{toeplitz}
\end{equation}
and the $(NP-L)$-vector $\bm{e}_N$ is defined as
\begin{equation} \bm{e}_N \triangleq [\begin{array}{cccc}
[\bm{e}(0)]_{L+1:P}^T & \bm{e}(1)^T & \ldots & \bm{e}(N-1)^T
\end{array}]^T. \label{eq:noise_eN}
\end{equation}
The product of matrices $\bm{G}$ and $\bm{H}$ defined above is actually equivalent
to a matrix in a ``fat" Toeplitz form as in (\ref{fatToeplitz}). The
reason we use this seemingly more complicated expression here is for
convenience of CRB derivation, as will be shown later in the next
section.

%which is to discard the first and the last $L$ elements of the
%channel filter output. The reason is that, the first $L$ entries of
%$\bm{s}_N$ will contain the inter-block interference (IBI) from
%prior blocks, and the last $L$ entries of $\bm{s}_N$ will be the IBI
%to posterior blocks, and thus it will overlap with the CP part of
%the following block. The channel matrix $H$ is of the same form as
%in (\ref{toeplitz}) but with a larger size $(NP+L) \times NP$. The
%last term $\bm{e}_N$ is still AWGN with zero mean and covariance
%matrix $\sigma^2 \bm{I}$.

Now, the goal of blind channel estimation in a redundant block
transmission system is to estimate
\begin{equation}
\bm{h} \triangleq \left[ \begin{array}{cccc} h_0 & h_1 & \ldots &
h_L
\end{array} \right]^T
\end{equation}
from the observation $\bm{y}_N$ defined in (\ref{main_eq}) and
(\ref{main_eq1}).

Although the problem studied here considers  block transmission
systems with all kinds of linearly redundant precoding including ZP
systems, it should be noted that even when $\bm{R}$ is taken as the
form in (\ref{eq:Rzp}), the problem is slightly different from that
defined in \cite{Barbarossa2002} due to the fact that the first $L$
entries of $\bm{y}(0)$ are not taken as part of the observation
vector $\bm{y}_N$.

\section{CRB for Blind Channel Estimators} \label{sec_crb_blind}
We define the $(NM+L+1)$-fold parameter vector $\bm{\theta}$ as
\begin{equation}
\bm{\theta} \triangleq \left[\begin{array}{cc} \bm{h}^T & \bm{s}_N^T
\end{array} \right]^T,
\end{equation}
where $\bm{s}_N$ is treated as a nuisance parameter
\cite{VanTrees2007}. To derive the CRB for blind channel estimators,
we first identify the probability density function of the
observation vector $\bm{y}_N$ given $\boldsymbol{\theta}$ and
calculate the complex Fisher information matrix (FIM) $\bm{J}$
defined in (\ref{eq_def_fim}). Then we apply Theorem
\ref{them_complex_constrained_crb} to obtain the CRB expression in
terms of $\bm{F}, \bm{h}, $ and $\bm{s}_N$.

By the assumption of additive white Gaussian noise, the probability
density function of the receiver's observation given $\bm{\theta}$
is
\begin{align}
&p(\bm{y}_N; \bm{\theta}) \notag \\
&\quad = \frac{1}{\left( \pi\sigma \right)^{NP-L}} \exp \left(
-\frac{1}{\sigma^2} \left\Vert \bm{y}_N - \bm{G}\bm{H}\left(
\bm{I}_N \otimes \bm{F} \right) \bm{s}_N \right\Vert^2 \right).
\end{align}
The Toeplitz matrix $\bm{H}$ can be rewritten as
\begin{equation}
\bm{H} = \sum_{l=0}^L h_l\bm{J}_l,
\end{equation}
where the $(i,j)$th element of $\bm{J}_l$ is defined as
\begin{equation}
[\bm{J}_l]_{i,j} = \left\{ \begin{array}{ll}
1, & \mbox{ if } i-j = l \\
0, & \mbox{ otherwise}
\end{array}, \right. l=0, 1, ..., L.
\end{equation}

In the following derivations, for simplicity, we will use $\bm{K}$
and $\bm{K}_l, l=0, 1, ..., L,$ to represent $ \bm{K} \triangleq
\bm{G}\bm{H}(\bm{I}_N\otimes \bm{F})$ and $ \bm{K}_l \triangleq
\bm{G}\bm{J}_l(\bm{I}_N\otimes \bm{F}), l=0, 1, ..., L,$
respectively. Note that $ \bm{K} = \sum_{l=0}^L h_l\bm{K}_l.$ To
calculate the FIM $\bm{J}$, we take partial derivatives on the
logarithm probability density function with respect to elements of
$\bm{\theta}$ and obtain the following equations.
\begin{align}
\frac{\partial \ln p}{\partial h^*_l} &= \dfrac{1}{\sigma^2} \bm{s}_N^H (\bm{I}_N \otimes \bm{F})^H \bm{J}_l^H \bm{G}^H \notag \\
&\quad \left[ \bm{y}_N - \bm{G}\bm{H}(\bm{I}_N \otimes \bm{F}) \bm{s}_N \right] \notag \\
&= \dfrac{1}{\sigma^2} \bm{s}_N^H \bm{K}_l^H \bm{e}_N.
\label{dev_1_1}
\end{align}
\begin{align}
\frac{\partial \ln p}{\partial h_l} &= \left( \frac{\partial \ln p}{\partial h^*_l} \right)^*. \label{dev_1_2} \\
\frac{\partial \ln p}{\partial \bm{s}_N^*} &= \frac{1}{\sigma^2} (\bm{I}_N \otimes \bm{F})^H \bm{H}^H \bm{G}^H \notag \\
&\quad \left[ \bm{y}_N - \bm{G}\bm{H}(\bm{I}_N \otimes \bm{F}) \bm{s}_N \right] \notag \\
&= \frac{1}{\sigma^2} \bm{K}^H \bm{e}_N. \label{dev_1_3} \\
\frac{\partial \ln p}{\partial \bm{s}_N} &= \left( \frac{\partial \ln p}{\partial \bm{s}_N^*} \right)^* \label{dev_1_4}
\end{align}
 Now, divide $\bm{J}$ into $4$ sub-matrices
\begin{align}
\bm{J} = \mathsf{E}\left[ \left( \frac{\partial \ln p}{\partial
\bm{\theta^*}} \right) \left( \frac{\partial \ln p}{\partial
\bm{\theta^*}} \right)^H \right] \triangleq \left[
\begin{array}{cccc}
\bm{J}_{0,0} & \bm{J}_{0,1} \\
\bm{J}_{1,0} & \bm{J}_{1,1} \\
\end{array} \right].
\end{align}
The entries of these sub-matrices can be written as
\begin{align}
[\bm{J}_{0,0}]_{i,j} &\triangleq \mathsf{E} \left[ \frac{\partial \ln p}{\partial h^*_i} \frac{\partial \ln p}{\partial h_j} \right], \\
[\bm{J}_{0,1}]_{i,j} &\triangleq \mathsf{E} \left[ \frac{\partial \ln p}{\partial h^*_i} \frac{\partial \ln p}{\partial [\bm{s}_N]_j} \right], \\
[\bm{J}_{1,0}]_{i,j} &\triangleq \mathsf{E} \left[ \frac{\partial
\ln p}{\partial [\bm{s}_N^*]_i} \frac{\partial \ln p}{\partial h_j}
\right],
\end{align} and
\begin{align}
[\bm{J}_{1,1}]_{i,j} &\triangleq \mathsf{E} \left[ \frac{\partial
\ln p}{\partial [\bm{s}_N^*]_i} \frac{\partial \ln p}{\partial
[\bm{s}_N]_j} \right],
\end{align} respectively. By equations (\ref{dev_1_1}, \ref{dev_1_2}, \ref{dev_1_3},
\ref{dev_1_4}) and the fact that $\mathsf{E}[\bm{e}_N\bm{e}_N^H] =
\sigma^2\bm{I}_{NP-L} $, we have the following results.
\begin{enumerate}
\item For the sub-matrix $\bm{J}_{0,0}$:
\begin{equation}
[\bm{J}_{0,0}]_{i,j} = \frac{1}{\sigma^2} \bm{s}_N^H \bm{K}_i^H
\bm{K}_j \bm{s}_N.
\end{equation}
Therefore, we can write $\bm{J}_{0,0}$ as
\begin{align}
\bm{J}_{0,0} &= \frac{1}{\sigma^2}
\left[\begin{array}{c}
\bm{s}_N^H\bm{K}_0^H \\
\vdots \\
\bm{s}_N^H\bm{K}_L^H
\end{array}\right]
\left[\begin{array}{c}
\bm{s}_N^H\bm{K}_0^H \\
\vdots \\
\bm{s}_N^H\bm{K}_L^H
\end{array}\right]^H. \label{j_00}
\end{align}

\item For the sub-matrix $\bm{J}_{0,1}$, the $i$th row of $\bm{J}_{0,1}$ is
\begin{align}
\mathsf{E}\left[ \frac{\partial \ln p}{\partial h^*_i} \left(
\frac{\partial \ln p} {\partial \bm{s}_N}^T \right) \right] &=
\frac{1}{\sigma^2} \bm{s}_N^H \bm{K}_i^H\bm{K}.
\end{align}
Therefore,
\begin{equation}
\bm{J}_{0,1} = \frac{1}{\sigma^2}
\left[ \begin{array}{c}
\bm{s}_N^H \bm{K}_0^H\bm{K} \\
\vdots \\
\bm{s}_N^H \bm{K}_L^H\bm{K}
\end{array} \right]. \label{j_01}
\end{equation}

\item For the sub-matrix $\bm{J}_{1,0}$, we have
\begin{equation}
\bm{J}_{1,0} = \bm{J}_{0,1}^H. \label{j_10}
\end{equation}

\item For the sub-matrix $\bm{J}_{1,1}$, we have
\begin{align}
\bm{J}_{1,1} &= \mathsf{E}\left[ \frac{\partial \ln p}{\partial
\bm{s}_N^*} \left( \frac{\partial \ln p}{\partial \bm{s}_N^*}
\right)^H \right] = \frac{1}{\sigma^2} \bm{K}^H\bm{K}. \label{j_11}
\end{align}
\end{enumerate}

Since $\bm{\theta} = [ \begin{array}{cc} \bm{h}^T & \bm{s}_N^T
\end{array} ]^T$ and $\bm{\theta}' \triangleq [ \begin{array}{cc}
(1/c) \bm{h}^T & c \bm{s}_N^T \end{array} ]^T$ will result in the
same observation vector $\bm{y}$ for any nonzero $c \in \mathbb{C}$,
an unbiased blind channel estimator will not exist. This is commonly
known as the scalar ambiguity issue and is usually resolved by
setting one element in $\bm{\theta}$ as already known
\cite{Barbarossa2002}.  Here we follow the approach in
\cite{Barbarossa2002} that we assume the $d$th element of $\bm{h}$,
$h_d$, is nonzero and already known to the receiver. In the
perspective of CRB for constrained parameters, this assumption
corresponds to the constraint function
\begin{equation}
f( \bm{\theta} ) = h_d - h_d^0,
\end{equation}
where $h_d^0$ is the realized value of $h_d$. Obviously $f$ is a
holomorphic map. We can then apply Theorem
\ref{them_complex_constrained_crb} and obtain the CRB with the
constraint function $f$ as
\begin{equation}
\mathsf{cov}(\boldsymbol{\theta}, \boldsymbol{\theta}) \geq
\bm{E}_d'^H \left( \bm{E}_d' \bm{J} \bm{E}_d'^H \right)^{-1}
\bm{E}_d', \label{eq_constrained_crb_blind_estimation}
\end{equation}
where $\bm{E}_d'$ is obtained by removing the $d$th row of
$\bm{I}_{NM+L+1}$, the identity matrix of order $(NM+L+1)$ (Note
that $\bm{E}_d'$ satisfies (\ref{eq:orthonormalU})). All elements on
the $d$th row and $d$th column of
(\ref{eq_constrained_crb_blind_estimation}) are zero, implying the
$d$th parameter, $h_d$, always has a zero MSE, consistent with the
fact that the receiver has {\it a priori} information
$\hat{h_d} = h_d^0$. Therefore we ignore the $d$th row and column of
(\ref{eq_constrained_crb_blind_estimation}) and only consider the
covariance bound for estimators of
$\boldsymbol{\theta}_d\triangleq[h_0, \ldots, h_{d-1}, h_{d+1},
\ldots, h_L, \bm{s}_N^T]^T$, (i.e., $\left( \bm{E}_d' \bm{J}
\bm{E}_d'^H \right)^{-1} $) in the following discussion. Let
$\tilde{\bm{J}} \triangleq \bm{E}_d' \bm{J} \bm{E}_d'^H
 $. Then the matrix $\tilde{\bm{J}}$ is of the form
\begin{equation} \label{eq_tilde_J}
\tilde{\bm{J}} \triangleq
\left[ \begin{array}{cc}
\bm{E}_d \bm{J}_{0,0} \bm{E}_d^H & \bm{E}_d\bm{J}_{0,1} \\
\bm{J}_{0,1}^H \bm{E}_d^H & \bm{J}_{1,1}
\end{array} \right]
\end{equation}
where $\bm{E}_d$ is an $L\times (L+1)$ matrix obtained by removing
the $d$th row from $\bm{I}_{L+1}$. Then, equation
(\ref{eq_constrained_crb_blind_estimation}) implies that for any
unbiased estimator $\hat{\bm{\theta}_d}$ of $\bm{\theta}_d$,
\begin{equation}
\mathsf{cov}\left( \hat{\bm{\theta}_d}, \hat{\bm{\theta}_d} \right)
\geq \tilde{\bm{J}}^{-1}. \label{eq:covth_geq_Jinv}
\end{equation}
Since we only focus on the performance of channel estimators, we can
simplify the above equation to obtain a bound for $\mathsf{cov}(
\hat{\bm{h}_d}, \hat{\bm{h}_d} )$, where $\bm{h}_d \triangleq [h_0,
..., h_{d-1}, h_{d+1}, ..., h_L]$. This can be done by noting that
(\ref{eq:covth_geq_Jinv}) implies
\begin{align}
&\left[ \begin{array}{cc} \bm{I}_L & \bm{0}_{L\times NM} \end{array} \right] \mathsf{cov}\left( \hat{\bm{\theta}_d},
 \hat{\bm{\theta}_d} \right) \left[ \begin{array}{c} \bm{I}_L^H \\ \bm{0}_{L\times NM}^H \end{array} \right] \geq \notag \\
&\quad\quad\quad\quad\quad\quad \left[ \begin{array}{cc} \bm{I}_L &
\bm{0}_{L\times NM} \end{array} \right] \tilde{\bm{J}}^{-1} \left[
\begin{array}{c} \bm{I}_L^H \\ \bm{0}_{L\times NM}^H \end{array}
\right].
\end{align}
The left side of the inequality corresponds to the covariance matrix
of $\bm{h}_d$. Therefore, the CRB for any unbiased estimator
$\hat{\bm{h}_d}$, of $\bm{h}_d$, can be calculated as the upper-left
$L$-by-$L$ sub-matrix of $\tilde{\bm{J}}^{-1}$:
\begin{equation}
\bm{C}_\mathrm{CRB} \triangleq \left[ \bm{E}_d \left( \bm{J}_{0,0} -
\bm{J}_{0,1} \bm{J}_{1,1}^{-1} \bm{J}_{0,1}^H \right) \bm{E}_d^H
\right]^{-1}. \label{eq:crlb}
\end{equation}
\begin{figure*}[t!]
\normalsize
\begin{equation}
\bm{C}_\mathrm{CRB}(\bm{h}, \bm{s}_N, \bm{F}, d) = \sigma^2 \left(
\bm{E}_d \ \tilde{\mathcal{U}} \left\{ \bm{I}_{(N-1)L} \otimes
\left[ (\bm{I}_N \otimes \bm{F})^* \bm{s}_N^* \bm{s}_N^T (\bm{I}_N
\otimes \bm{F})^T \right] \right\} \tilde{\mathcal{U}}^H \
\bm{E}_d^H \right)^{-1}. \label{eq_crb_blockTrans}
\end{equation}
\hrulefill \vspace*{4pt}
\end{figure*}
From Appendix \ref{appendix_proof_of_crb}, we obtain the final form
of the CRB for blind channel estimation in any given redundant block
transmission systems, as shown in (\ref{eq_crb_blockTrans}), in
terms of $\bm{h}$, $\bm{s}_N$, $d$, and $\bm{F}$. In
(\ref{eq_crb_blockTrans}), the matrix $\tilde{\cal U}$ has a size of
$(L+1)\times LPN(N-1)$ and can be determined by vector $\bm{h}$ and
matrix $\bm{F}$. The values of nonzero entries of $\tilde{{\cal U}}$
all come from the $(N-1)L$ orthonormal left annihilators of $\bm{K}
= \bm{G}\bm{H}(\bm{I}_N\otimes \bm{F})$ and are defined in details
in (\ref{eq:Utilde}).

Although the CRB considers all LRP, including ZP systems (i.e.,
$\bm{F}$ satisfies (\ref{eq:F=RF_tilde})(\ref{eq:Rzp})), the CRB
expression in (\ref{eq_crb_blockTrans}) does not reduce to that
obtained in \cite{Barbarossa2002, Su2008a} in this situation. The
reason is that in the derivation of this CRB, we do not use the
first $L$ entries of the first block as part of the observation
vector. In fact, as will be shown in Section \ref{sec_num_results},
with a smaller observation vector, the CRB values derived here are
expected to be slightly higher than those derived in
\cite{Barbarossa2002, Su2008a}.

\section{Numerical Results} \label{sec_num_results}
In this section we first compare the performance of existing blind
channel estimation algorithms for CP systems with the CRBs derived
in this paper. The subspace-based blind channel estimators proposed
in \cite{Su2007, Muquet2002} are chosen as benchmarks.

The simulation settings are concluded in Table
\ref{table_sim_setting}. The simulation results are averaged over
$N_{ch} = 500$ independent channel realizations. All elements in the
$i$th realization of $\bm{h}$, defined as $\bm{h}^{(i)} =
[\begin{array}{cccc}h_0^{(i)}&h_1^{(i)}&\cdots&h_L^{(i)}\end{array}]$,
are generated first as independent complex Gaussian random variables
with zero mean and unity variance and then normalized to satisfy
$\Vert\bm{h}^{(i)}\Vert^2=1$. For each channel realization, $N_s =
30$ trials with independent sets of QPSK modulated data symbols are
conducted to further average the channel estimation MSE, defined as:

\begin{equation}
\mbox{MSE} \triangleq {1\over
N_{ch}N_{s}}\sum_{i=1}^{N_{ch}}\sum_{j=1}^{N_{s}} \Vert
(h_{d_i}^{(i)} / \hat{h}^{(i,j)}_{d_i}) \hat{\bm{h}}^{(i,j)} -
\bm{h}^{(i)} \Vert^2,
\end{equation}
where $\hat{\bm{h}}^{(i,j)}$ is the estimate of the $i$th channel
realization in the $j$th trial, and $d_i$ is chosen as the index of
the channel tap with the maximal power: $d_i = \arg\max_{0\leq d\leq
L} \left|h_d^{(i)}\right|^2.$ Similarly, the CRB curves are
calculated according to
\begin{equation}
\mbox{CRB} \triangleq {1\over N_{ch}N_{s}}
\sum_{i=1}^{N_{ch}}\sum_{j=1}^{N_{s}}
\mathsf{tr}\left[{\bm{C}_\mathrm{CRB}(\bm{h}^{(i)},
\bm{s}_N^{(i,j)}, \bm{F}, d_i)}\right]
\end{equation} with the same channel and data set realizations.

%Since blind channel estimation algorithms can only identify the
%channel response up to a scalar ambiguity, the normalized squared
%error for the algorithm output is calculated as follows. Let the
%real channel value be $\bm{h}$ and the algorithm output be
%$\hat{\bm{h}}$. We calculate the normalized squared error as
%\begin{equation}
%\mbox{squared error} \triangleq \frac{\Vert (h_d / \hat{h}_d) \hat{\bm{h}} - \bm{h} \Vert^2}{\Vert \bm{h} \Vert^2},
%\end{equation}

\begin{table}[!t]
\renewcommand{\arraystretch}{1.3}
\caption{Simulation Settings for CP Systems}
\label{table_sim_setting} \centering
\begin{tabular}{|c|c|}
\hline
Block size $M$ & 12 \\
\hline
Number of received blocks $N$ & 8, 25, 50 \\
\hline
Channel order $L$ & 4 \\
\hline
Modulation & QPSK \\
\hline
\end{tabular}
\end{table}

\begin{figure}[!t]
\centering
\includegraphics[width=3.5in]{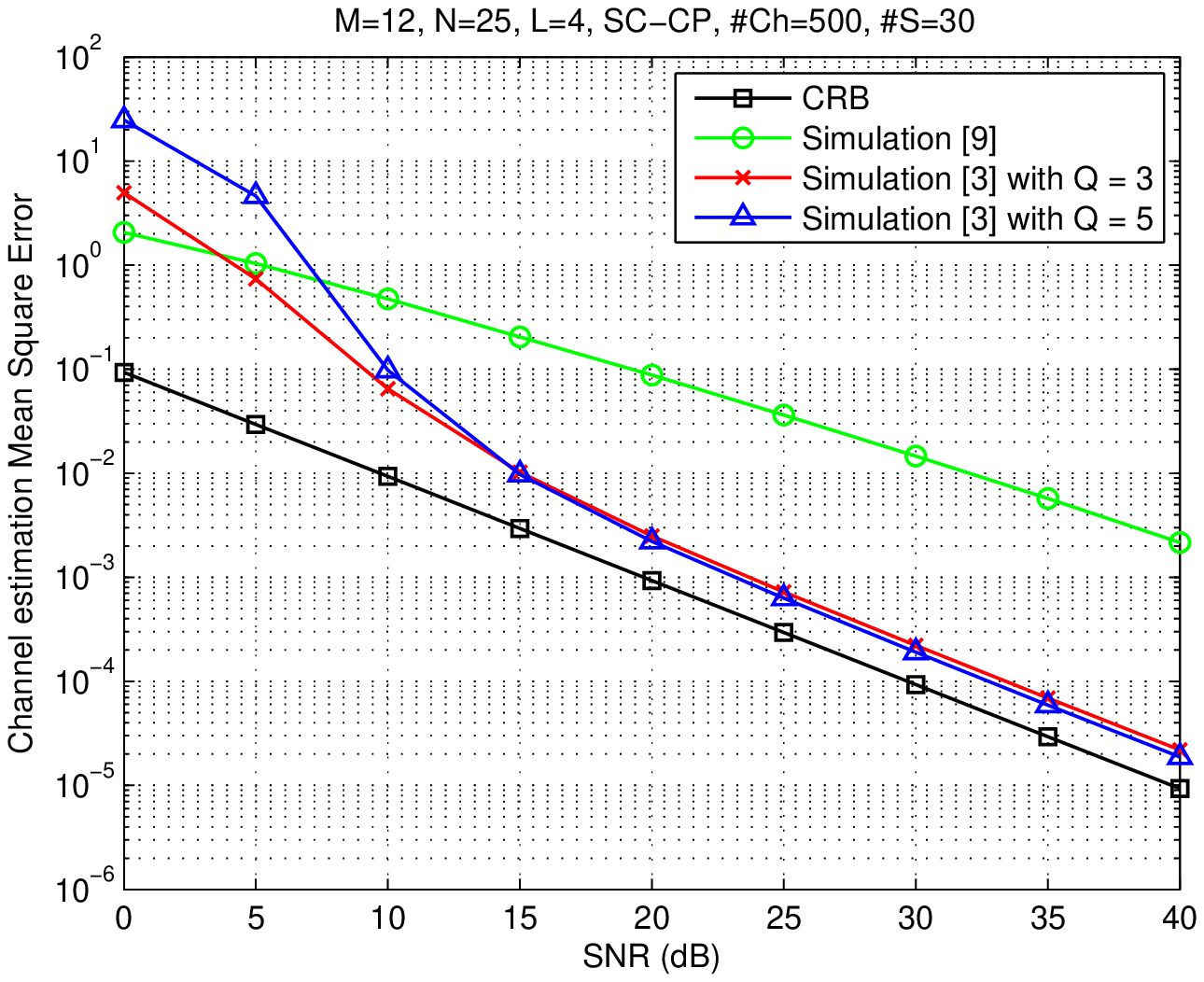}
\caption{CRB and simulatoin results for SC-CP system with 25 blocks}
\label{result_sccp_J25}
\end{figure}

\begin{figure}[!t]
\centering
\includegraphics[width=3.5in]{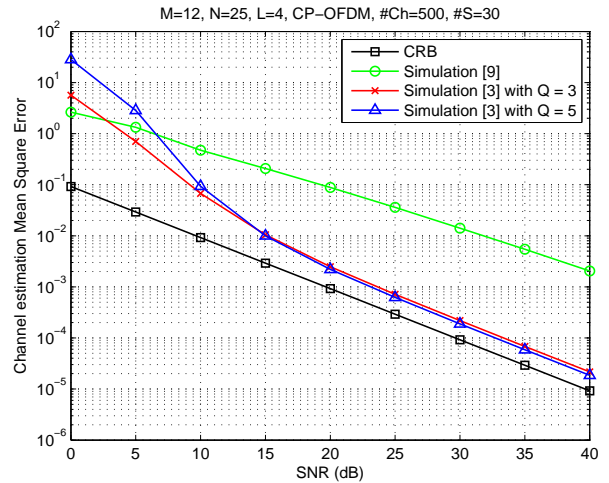}
\caption{CRB and simulatoin results for OFDM-CP system with 25
blocks} \label{result_ofdm_J25}
\end{figure}

\begin{figure}[!t]
\centering
\includegraphics[width=3.5in]{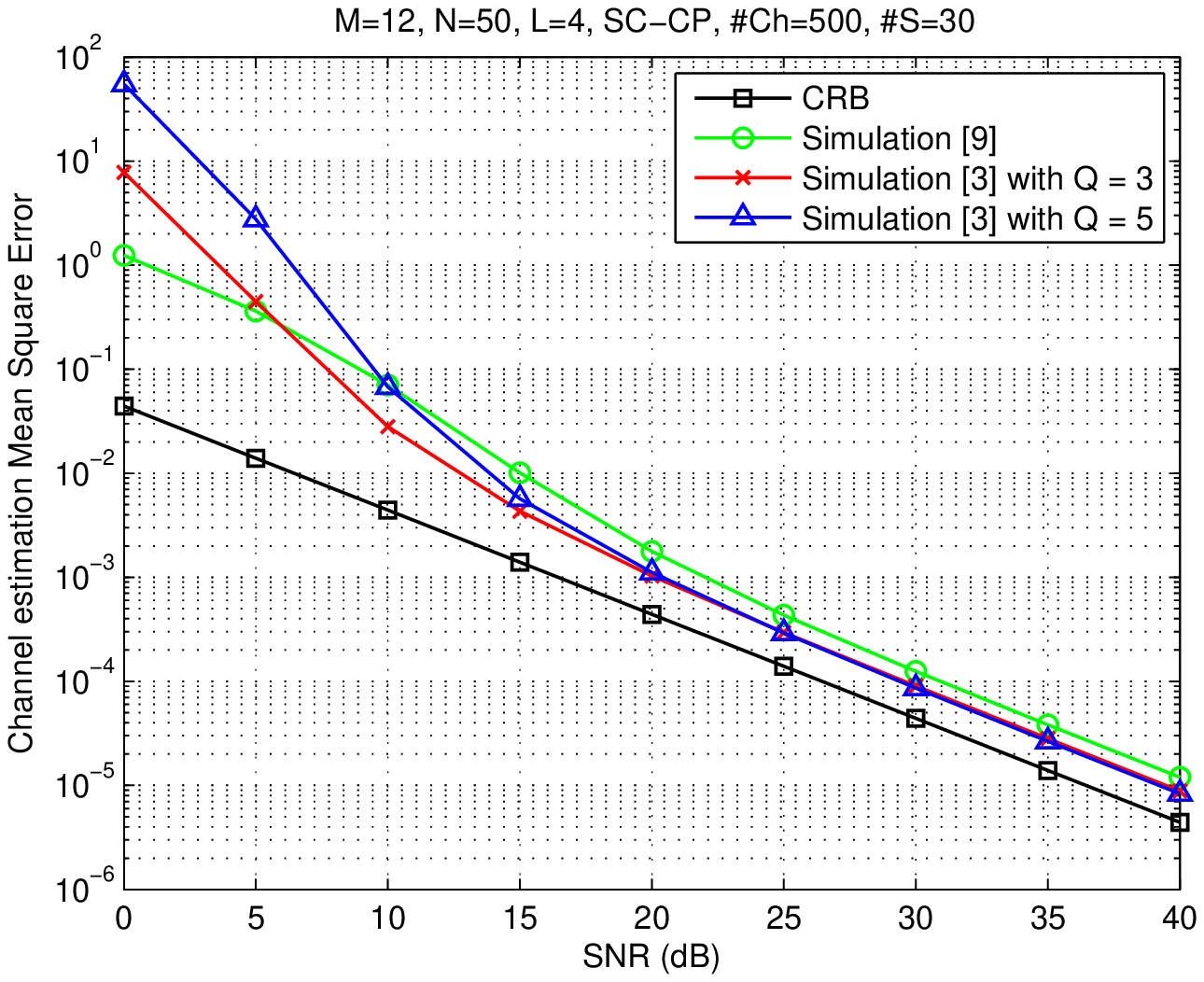}
\caption{CRB and simulatoin results for SC-CP system with 50 blocks}
\label{result_sccp_J50}
\end{figure}

\begin{figure}[!t]
\centering
\includegraphics[width=3.5in]{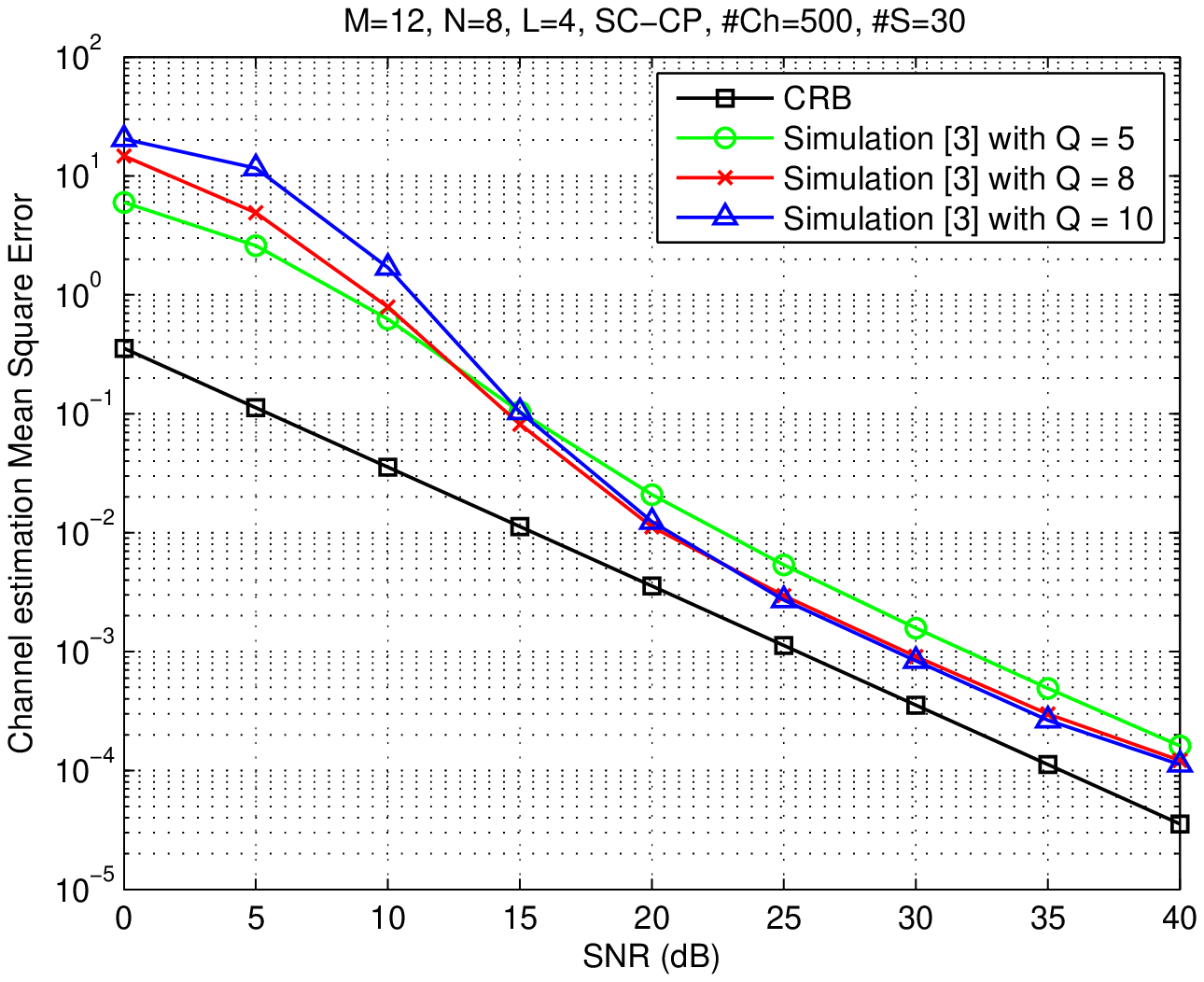}
\caption{CRB and simulatoin results for SC-CP system with 8 blocks}
\label{result_sccp_J8}
\end{figure}

\begin{figure}[!t]
\centering
\includegraphics[width=3.5in]{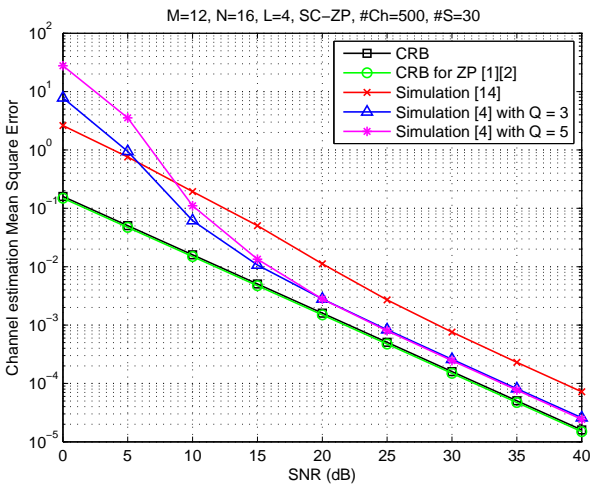}
\caption{CRB and simulatoin results for SC-ZP system}
\label{result_crb_sczp}
\end{figure}

We first consider the blind channel estimation problem in SC-CP
($\tilde{\bm{F}} = \bm{I}_M$, $\bm{R}$ is chosen as in
(\ref{eq:Rcp})) systems with 25 received blocks available. The
performance of algorithms proposed in \cite{Su2007} and
\cite{Muquet2002}, as well as the CRB derived in this paper, are
depicted in Figure \ref{result_sccp_J25}. We can see the algorithm
in \cite{Su2007} has advantage over that in \cite{Muquet2002} in
high-SNR region while the algorithm in \cite{Muquet2002} has better
performance with a low SNR. As predicted in \cite{Su2007},
increasing the algorithm parameter $Q$ from 3 to 5 gains a slight
MSE improvement in high SNR region but results in a large
performance degradation when SNR is low. Under high SNR values, for
both curves of algorithms\cite{Su2007}, the gap between the
simulation MSE result and the CRB, tends to approach a constant in
log scale (around 3 to 4 dB in SNR). All these MSE results do not
achieve the performance lower bound suggested by the CRB.

Figure \ref{result_ofdm_J25} considers CP-OFDM systems (i.e.,
$\tilde{\bm{F}} = \bm{W}^H$)with the same simulation settings. The
performance curves are almost identical to those in Figure
\ref{result_sccp_J25}. This suggests the precoder $\tilde{\bm{F}}$
does not significantly affect the MSE and CRB curves for blind
channel estimators in CP systems. All following simulations will
then consider single carrier systems only.

In Figure \ref{result_sccp_J50}, we consider the problem in SC-CP
systems with 50 received blocks. The MSE and CRB curves are all
lower than those in previous plots. The gaps between the CRB and MSE
results still converge to constant values in log scale in the
high-SNR region. For both curves of algorithm \cite{Su2007}, the gap
is around 3 dB; for algorithm \cite{Muquet2002}, it is around 5dB.

The case with only 8 available received blocks is studied in Figure
\ref{result_sccp_J8}. As algorithm in \cite{Muquet2002} does not
work appropriately with such a small amount of received data, only
performance curves of algorithm in \cite{Su2007}, with large values
of $Q$, are presented. We observe that performance in the high-SNR
region improves when the algorithm parameter $Q$ is chosen as a
large value. But the improvement slows down gradually as $Q$
increases to 10. Similarly, it results in performance degradation
with low SNR values. For the high SNR region, in the best situation
shown ($Q=10$), the gap between the MSE performance and the CRB is
around 4 to 5 dB in high SNR region.

We now turn our attention to ZP systems (i.e., $\bm{F}$ satisfies
(\ref{eq:F=RF_tilde})(\ref{eq:Rzp})). We compared the derived CRB,
CRB reported in \cite{Barbarossa2002, Su2008a}, and MSE performances
of existing blind algorithms \cite{Su2007a, Scaglione1999b} in
Figure \ref{result_crb_sczp}. In Figure \ref{result_crb_sczp}, our
CRB is slightly higher than that reported in \cite{Barbarossa2002,
Su2008a}, with a margin of less than 0.5dB. This is consistent with
the fact that the first $L$ entries of the first block are dropped
from the observation vector (as stated at the end of Sections
\ref{sec_sys_model} and \ref{sec_crb_blind}). This discrepancy in
two CRBs, however, is much smaller than the gap between MSE
performances of existing algorithms and the CRBs. The generalized
algorithm in \cite{Su2007a} has a performance gap to the CRB of
around 2 dB in high SNR region. The performance curve of an earlier
method in \cite{Scaglione1999b} is around 7dB from the CRB bound.

Summarizing all above simulation results, in both CP and ZP systems,
all existing blind channel estimation algorithms \cite{Muquet2002,
Su2007, Scaglione1999b, Su2007a} do not achieve the lower bound
suggested by CRB. In the high-SNR region, the generalized subspace
methods proposed in \cite{Su2007a} \cite{Su2007} obtain the best
performances among others, regardless of the number of available
received blocks. But a gap of around 2 to 4 dB from the best
performance to the CRB is constantly present, suggesting there is
still room for performance improvement of blind channel estimators
in these systems.

%As the theoretical results indicate, the derived CRB is a lower
%bound for the simulation result. The simulation result approaches
%the CRB as the SNR increases under low SNR values, and the gap
%between the simulation results and and the CRB becomes unchanged
%under high SNR values in log scale. The existing blind channel
%estimation algorithms do not achieve the CRB\footnote{In
%\cite{Su2007}, it is shown that as the value of $Q$ increases, the
%MSE will decrease in high SNR range. The performance in low SNR
%range, however, will become much worse, so only a moderate value of
%$Q$ is meaningful for practical use.}.
%

\section{Conclusions and Future Work} \label{sec_conclusions}
In this paper, we derived the CRB for blind channel estimators in
redundant block transmission systems. The derived CRB is valid for
any full-rank linear redundant precoder (LRP), including both ZP and
CP systems and is an extension to the result in
\cite{Barbarossa2002,Su2008a}. We compared the derived CRB with the
performances of the existing subspace-based blind channel estimators
for CP \cite{Su2007, Muquet2002} and ZP systems
\cite{Scaglione1999b,Su2007a}. Numerical results show that these
existing blind channel estimators for ZP and CP systems do not
achieve the CRB and there is still some room for performance
improvement.

In addition, a simple form of CRB formula is derived for
unconstrained and constrained complex parameters. We extended this
form, originally for \emph{real}
parameters\cite{Gorman1990,Stoica1998,VanTrees2001}, to the case of
multiple \emph{complex} parameters. The results not only facilitate
the derivation of CRB for the blind channel estimation problem of
interest, but also are expected to be useful for other complex CRB
derivations.

In the future, there are a few research directions worthy of being
explored. The CRB derived in this paper is applicable to blind
channel estimators which do not have \textit{a priori} information
about the transmitted signals\footnote{We ignore the pilot symbol
required to eliminate the scalar ambiguity here.}. It would be
desirable to derive the CRB for blind channel estimator that assume
some stronger assumptions about the transmitted signals. For
example, the receiver should know the modulation scheme the
transmitter uses, which corresponds to the finite alphabet
assumption \cite{Zhou2001}. One possible approach for this is to use
the CRB for constrained complex parameters. Take quadrature
phase-shift keying (QPSK) for example. The finite-alphabet
assumption that QPSK is used can be modeled by setting the
constraint function as $\bm{f}( \bm{\theta} ) = [
\begin{array}{cccc} s_0(0)^4 & s_1(0)^4 & \ldots &
s_{M-1}(N-1)^4 \end{array} ]^T - [
\begin{array}{cccc} 1 & 1 & \ldots & 1 \end{array} ]^T$, assuming
each modulation symbol has unit power. It is obvious that $\bm{f}$
is holomorphic, so Theorem \ref{them_complex_constrained_crb} can be
applied. This approach, however, cannot be directly extended to more
complex modulation schemes such as 16-quadrature amplitude
modulation (16-QAM). And this will require additional research
efforts.

Furthermore, the receiver may, in additional to the modulation
scheme, also know the joint probability distribution of the
transmitted symbols. The CRB for this case, at first glance, can be
derived by Bayesian Cram\'{e}r-Rao bound (BCRB), or Van Trees
inequality \cite{VanTrees2001,Gill1995}. BCRB is a lower bound of
variance for all estimators knowing the \textit{a priori}
distribution of the unknown parameters. The BCRB, however, cannot be
applied to discrete parameters, such as modulation symbols. An
extension of the original BCRB formula to discrete parameters or a
new bound for such cases is desirable.

\appendices
\section{Proof of Theorem \ref{thm_complex_crlb}} In order to prove Theorem
\ref{thm_complex_crlb}, we apply the following lemma analogous to
the well-known Cauchy-Schwartz inequality, which can be considered
as the Cauchy-Schwartz inequality for random vectors. Although this
lemma has been already proved in \cite{Gorman1990,Rao2000}, we show
the proof here because it helps illustrate the necessary and
sufficient condition of when the equality of CRB holds.

\begin{lem} \label{thm_var_cov_ineq}
For any two random vectors $\boldsymbol{x}, \boldsymbol{y} \in
\mathbb{F}^n$, where $\mathbb{F}$ can be $\mathbb{R}$ or
$\mathbb{C}$,
\begin{equation}
\mathsf{cov}( \boldsymbol{y}, \bm{y} ) \geq \mathsf{cov}(
\boldsymbol{y}, \boldsymbol{x} ) \mathsf{cov}( \boldsymbol{x},
\bm{x} )^{\dagger} \mathsf{cov}( \boldsymbol{x}, \boldsymbol{y} ).
\label{var_cov_ineq}
\end{equation}
\end{lem}
\begin{proof}
The proof presented here is from \cite{Rao2000}\footnote{The theorem
in \cite{Rao2000} is stronger than the version presented here since
it allows the Moore-Penrose generalized inverse to be substituted by
generalized inverse \cite{Rao1972}.}. Define
\begin{equation}
\begin{array}{ll}
\boldsymbol{\Sigma}_{11} \triangleq \mathsf{cov}(\boldsymbol{x}, \bm{x}), & \boldsymbol{\Sigma}_{12} \triangleq \mathsf{cov}(\boldsymbol{x}, \boldsymbol{y}), \\
\boldsymbol{\Sigma}_{21} \triangleq \mathsf{cov}(\boldsymbol{y},
\boldsymbol{x}), & \boldsymbol{\Sigma}_{22} \triangleq
\mathsf{cov}(\boldsymbol{y}, \bm{y}).
\end{array}
\end{equation}
Define a random vector $\boldsymbol{z} \triangleq \boldsymbol{y} -
\boldsymbol{\Sigma}_{21}\boldsymbol{\Sigma}_{11}^{\dagger}\boldsymbol{x}$,
then we have
\begin{align}
\mathsf{cov}(\boldsymbol{z}, \bm{z}) &= \boldsymbol{\Sigma}_{22} - \boldsymbol{\Sigma}_{21} \left( \boldsymbol{\Sigma}_{11}^{\dagger} \right) ^H \boldsymbol{\Sigma}_{21} ^H - \boldsymbol{\Sigma}_{21}\boldsymbol{\Sigma}_{11}^{\dagger}\boldsymbol{\Sigma}_{12} \\
&\quad + \boldsymbol{\Sigma}_{21} \boldsymbol{\Sigma}_{11}^{\dagger} \boldsymbol{\Sigma}_{11} \left( \boldsymbol{\Sigma}_{11}^{\dagger} \right) ^H \boldsymbol{\Sigma}_{21}^H \notag \\
&= \boldsymbol{\Sigma}_{22} -
\boldsymbol{\Sigma}_{21}\boldsymbol{\Sigma}_{11}^{\dagger}\boldsymbol{\Sigma}_{12}.
\end{align}
The second equality follows by $\left( \boldsymbol{A}^{\dagger}
\right)^H = \left( \boldsymbol{A}^H \right)^\dagger$ and
$\boldsymbol{A}^\dagger \boldsymbol{A} \boldsymbol{A}^\dagger =
\boldsymbol{A}^\dagger$ for any matrix $\boldsymbol{A}$.

The lemma follows since any covariance matrix is nonnegative
definite.
\end{proof}

Now we are ready to prove Theorem \ref{thm_complex_crlb}, a CRB for
unconstrained
complex parameters.\\

\noindent Proof of Theorem \ref{thm_complex_crlb}:
\begin{proof}
Since $\boldsymbol{t}$ is an unbiased estimator,
\begin{equation}
\mathsf{E}\left[ \bm{t}(\bm{y}) \right] = \boldsymbol{\theta}.
\end{equation}
Differentiate both sides with respect to $\bm{\theta}^T$, we can get
\begin{equation}
\int \bm{t}( \bm{y} ) \frac{\partial p}{\partial \bm{\theta}^T} \
d\bm{y} = \bm{I}_n;
\end{equation}
equivalently,
\begin{equation}
\int \bm{t}( \bm{y} ) \frac{\partial \ln p}{\partial \bm{\theta}^T}
p \ d\bm{y} = \bm{I}_n. \label{eq_E_T_U=I}
\end{equation}
With the regularity condition we can rewrite (\ref{eq_E_T_U=I}) as
\begin{equation}
\mathsf{cov}\left( \boldsymbol{t}, \frac{\partial \ln p}{\partial
\boldsymbol{\theta}^*} \right) = \boldsymbol{I}_n. \label{cov_T_U=I}
\end{equation}

Now, substitute $\boldsymbol{x}$ and $\boldsymbol{y}$ in
(\ref{var_cov_ineq}) by $(\partial \ln p) / (\partial
\boldsymbol{\theta}^*)$ and $\boldsymbol{t}$, respectively, we have
\begin{align}
&\quad \mathsf{cov}\left( \boldsymbol{t}, \bm{t} \right) \notag \\
&\geq \mathsf{cov}\left( \boldsymbol{t}, \frac{\partial \ln p}{\partial \boldsymbol{\theta}^*} \right) \mathsf{cov}\left( \frac{\partial \ln p}{\partial \boldsymbol{\theta}^*}, \frac{\partial \ln p}{\partial \boldsymbol{\theta}^*} \right)^{\dagger} \mathsf{cov}\left( \boldsymbol{t}, \frac{\partial \ln p}{\partial \boldsymbol{\theta}^*} \right)^H \notag \\
&=  \mathsf{cov}\left( \frac{\partial \ln p}{\partial \boldsymbol{\theta}^*}, \frac{\partial \ln p}{\partial \boldsymbol{\theta}^*} \right)^{\dagger} \notag \\
&= \mathsf{E} \left[ \left( \frac{\partial \ln p}{\partial
\boldsymbol{\theta}^*}\right) \left( \frac{\partial \ln p}{\partial
\boldsymbol{\theta}^*}\right)^H \right]^{\dagger} = \bm{J}^\dagger.
\end{align}
The first equality follows by (\ref{cov_T_U=I}), and the second
equality follows by the regularity condition.

As for the necessary and sufficient condition of the equality, from
the proof for Lemma \ref{thm_var_cov_ineq}, we have $\mathsf{cov}(
\bm{z}, \bm{z} ) = \bm{0}$ if and only if $\bm{z} = \bm{c}$ with
probability $1$ for some constant vector $\bm{c}$, that is,
\begin{equation}
\bm{t} - \bm{J}^\dagger \frac{\partial \ln p}{\partial
\boldsymbol{\theta}^*} = \bm{c}.
\end{equation}
Taking expectation on the both sides of the equality, by the
regularity condition, we have
\begin{equation}
\mathsf{E}\left[ \bm{t} \right] = \bm{c}.
\end{equation}
Since $\bm{t}$ is unbiased, $\bm{c} = \bm{\theta}$.\end{proof}

\begin{figure*}[!t]
\normalsize
% Store the current equation number.
\setcounter{MYtempeqncnt}{\value{equation}}
% Set the equation number to one less than the one
% desired for the first equation here.
% The value here will have to changed if equations
% are added or removed prior to the place these
% equations are referenced in the main text.
\setcounter{equation}{\value{magicEquations}}
\begin{align}
\bm{D} &= \frac{1}{\sigma^2} \left\{ \bm{I}_{L+1} \otimes \bm{x}_N^H
\right\} \left[\begin{array}{ccc}
\mathsf{vec}(\tilde{\mathcal{U}}_0) & \ldots &
\mathsf{vec}(\tilde{\mathcal{U}}_{(N-1)L-1})\end{array}\right]
 \left[\begin{array}{ccc} \mathsf{vec}(\tilde{\mathcal{U}}_0)
& \ldots &
\mathsf{vec}(\tilde{\mathcal{U}}_{(N-1)L-1})\end{array}\right]^H
\left\{ \bm{I}_{L+1} \otimes \bm{x}_N \right\}.
\label{eq_simp_D_first}
\end{align}
\hrulefill \vspace*{4pt}
\end{figure*}

\begin{figure*}[!t]
\normalsize
\begin{align}
\bm{D} &= \frac{1}{\sigma^2} \left[\begin{array}{ccc} \mathsf{vec}
\left( \bm{x}_N^H \tilde{\mathcal{U}}_0 \right) & \ldots &
\mathsf{vec}\left( \bm{x}_N^H \tilde{\mathcal{U}}_{(N-1)L-1} \right)
\end{array}\right]  \left[\begin{array}{ccc}
\mathsf{vec}\left( \bm{x}_N^H \tilde{\mathcal{U}}_0 \right) & \ldots
& \mathsf{vec}\left( \bm{x}_N^H \tilde{\mathcal{U}}_{(N-1)L-1}
\right)
\end{array}\right]^H. \label{eq_D_as_vec}
\end{align}
% Restore the current equation number.
\setcounter{equation}{\value{MYtempeqncnt}} \hrulefill \vspace*{4pt}
\end{figure*}
\section{Proof of Theorem \ref{them_complex_constrained_crb}}

To prove Theorem \ref{them_complex_constrained_crb}, we first derive
some lemmas for holomorphic maps. First of all, the lemma of
implicit functions is presented below.

\begin{lem} \label{lem:implicit_functions}
Let $B\subset \mathbb{C}^n \times \mathbb{C}^m$ be an open set,
$\bm{f}: B\to \mathbb{C}^m$ a holomorphic map, $( \bm{z}_0, \bm{w}_0
) \in B$ a point with $\bm{f}( \bm{z}_0, \bm{w}_0 ) = \bm{0}$, and
\begin{equation}
%\mathsf{det}\left( \frac{\partial \bm{f}}{\partial [z_{n+1}, \ldots,
%z_{n+m}]}(\bm{z}_0, \bm{w}_0) \right) \neq 0.
\mathsf{det}\left( \frac{\partial \bm{f}}{\partial [w_{1}, \ldots,
w_{m}]}(\bm{z}_0, \bm{w}_0) \right) \neq 0.
\end{equation}
Then there is an open neighborhood $U = U' \times U'' \subset B$ and
a holomorphic map $\bm{g}: U' \to U''$ such that
\begin{equation}
\{ ( \bm{z}, \bm{w} ) \in U'\times U'' : \bm{f}( \bm{z},
\bm{w} ) = \bm{0} \} = \{ ( \bm{z}, \bm{g}( \bm{z} ) ): \bm{z} \in
U' \}.
\end{equation}
\end{lem}

\begin{proof}
See \cite{Fritzsche2002}.
\end{proof}

Now consider the case where the possible values of $\bm{\theta}$ is
constrained to a set $\Theta \triangleq \left\{
\bm{\theta}\in\mathbb{C}^n: \bm{f}(\bm{\theta}) = \bm{0} \right\}$
defined by a holomorphic constraint function $\bm{f}: \mathbb{C}^n
\to \mathbb{C}^m$, $m\leq n$. Available observation is $\bm{y} \in
\mathbb{C}^p$ with pdf $p( \bm{y}; \bm{\theta} )$. The goal here is
to derive the CRB for any unbiased estimator $\bm{t}( \bm{y} )$ of
$\bm{\theta}$. Note now the unbiasedness refers to
\begin{equation}
\mathsf{E}\left[ \bm{t}( \bm{y} ) \right] = \bm{\theta}
\end{equation}
for all $\bm{\theta} \in \Theta$ instead of the whole
$\mathbb{C}^n$. Assume $\partial \bm{f} / \partial \bm{\theta}^T$
has full rank for all $\bm{\theta} \in \Theta$ so Lemma
\ref{lem:implicit_functions} applies. Then, we have the following
lemma.

\begin{lem} \label{thm_expectation_delta_equal_delta}
Let $\bm{t}( \bm{y} )$ be an unbiased estimator of $\bm{\theta}$ for
all $\bm{\theta} \in \Theta$. Then
\begin{equation}
\mathsf{E}\left[ \left( \bm{t} - \bm{\theta} \right) \frac{ \partial
\ln p( \bm{y}; \bm{\theta} )}{\partial \bm{\theta}^T} \right]
\bm{\delta} = \bm{\delta}
\end{equation}
for all $\bm{\theta} \in \Theta$, and $\bm{\delta}$ such that
\begin{equation}
\frac{\partial \bm{f}}{\partial \bm{\theta}^T} \bm{\delta} = \bm{0}.
\label{eq_property_delta}
\end{equation}
\end{lem}
\begin{proof}
Follow the derivation of \cite[Theorem I]{Marzetta1993}. Note that
the implicit function theorem and chain rule hold for holomorphic
constraint functions \cite{Fritzsche2002}.
\end{proof}

And a corollary immediately follows.

\begin{cor} \label{cor_expectation_U_equal_U}
Choose a matrix $\bm{U}$ with orthonormal columns such that
\begin{equation}
\frac{\partial \bm{f}}{\partial \bm{\theta}^T} \bm{U} = \bm{0}.
\end{equation}
Then
\begin{equation}
\mathsf{E}\left[ \left( \bm{t} - \bm{\theta} \right) \frac{ \partial
\ln p( \bm{y}; \bm{\theta} )}{\partial \bm{\theta}^T} \right] \bm{U}
= \bm{U}.
\end{equation}
\end{cor}
\begin{proof}
Note that every column of $\bm{U}$ satisfies
(\ref{eq_property_delta}).
\end{proof}
\vspace{0.3cm} Now we are ready to prove Theorem
\ref{them_complex_constrained_crb}.
\begin{proof}
Substitute $\boldsymbol{x}$ and $\boldsymbol{y}$ in
(\ref{var_cov_ineq}) by $\bm{U}^H (\partial \ln p) / (\partial
\boldsymbol{\theta}^*)$ and $\boldsymbol{t}$, respectively. Then
apply Corollary \ref{cor_expectation_U_equal_U} to simplify the
equation. As in the proof for Theorem \ref{thm_complex_crlb}, the
equality holds if and only if
\begin{equation}
\bm{t} - \bm{U} \left( \bm{U}^H \bm{J} \bm{U} \right)^\dagger
\bm{U}^H \frac{\partial \ln p}{\partial \bm{\theta}^*} = \bm{c}
\quad \mbox{with probability }1\label{eq_equality_hold_constrained}
\end{equation}
for some constant vector $\bm{c}$. Taking expectations on both sides
of (\ref{eq_equality_hold_constrained}) and we have $\bm{c} =
\bm{\theta}$.
\end{proof}

\section{Proof of CRB} \label{appendix_proof_of_crb}
In this appendix, we show how to derive the CRB from (\ref{eq:crlb})
to the final form (\ref{eq_crb_blockTrans}). Define $\bm{A}
\triangleq \bm{J}_{0,1} \bm{J}_{1,1}^{-1} \bm{J}_{0,1}^H$. Then,
using equations (\ref{j_01}) and (\ref{j_11}), the $(i,j)$th element
of $\bm{A}$ can be calculated as
\begin{align}
[\bm{A}]_{i,j} = &\frac{1}{\sigma^2} \bm{s}_N^H
\bm{K}_i^H\bm{K}\bm{K}^\dagger\bm{K}_j \bm{s}_N. \label{A_ij}
\end{align}
Define $\bm{D} \triangleq \bm{J}_{0,0} - \bm{J}_{0,1}
\bm{J}_{1,1}^{-1} \bm{J}_{0,1}^H$. Then we have, by (\ref{j_00}) and
(\ref{A_ij}),
\begin{align}\label{eq:D_ij}
[\bm{D}]_{i,j} = &{1\over\sigma^2}\bm{s}_N^H \bm{K}_i^H \left(
\bm{I}_{NP-L} - \bm{K}\bm{K}^\dagger \right) \bm{K}_j \bm{s}_N.
\end{align}

By assumption, $\tilde{\bm{F}}$ is a full-rank matrix, and thus
$\bm{F}$ is a full-column-rank matrix. The singular value
decomposition (SVD) of the matrix $\bm{K} = \bm{G}\bm{H}(\bm{I}_N
\otimes \bm{F})$, therefore, is of the form
\begin{align}
\quad \bm{K} = \bm{G} \bm{H} (\bm{I}_N \otimes \bm{F}) =
\left[\begin{array}{cc} \bar{\bm{U}} &
\tilde{\bm{U}}\end{array}\right] \left[\begin{array}{c} \bm{\Sigma}
\\ \bm{0} \end{array}\right] \bm{V}^H,
\end{align}
in which the matrix $\tilde{\bm{U}}$ represents the null space and
has a size of $(NP-L)\times (NL-L)$. Now, we have \begin{equation}
\bm{I}_{NP-L} - \bm{K}\bm{K}^\dagger = \tilde{\bm{U}}
\tilde{\bm{U}}^H,
\end{equation}
and we can rewrite $[\bm{D}]_{i,j}$ in (\ref{eq:D_ij}) as
\begin{equation}
[\bm{D}]_{i,j} = \frac{1}{\sigma_2}
\bm{x}_N^H\bm{J}_i^H\bm{G}^H\tilde{\bm{U}}(\bm{G}\tilde{\bm{U}})^H\bm{J}_j\bm{x}_N
\end{equation}
where \begin{equation}\label{eq:xN}\bm{x}_N \triangleq  \bm{s}_N^H
(\bm{I}_N \otimes \bm{F})^H
\end{equation} is the vector containing precoded transmitted
symbols. Now, we can express $\bm{D}$ as
\begin{align}
\bm{D} =& \frac{1}{\sigma^2} \left( \bm{I}_{L+1} \otimes \bm{x}_N^H
 \right) \left[\begin{array}{c}
\bm{J}_0^H \bm{G}^H \tilde{\bm{U}} \\
\vdots \\
\bm{J}_L^H \bm{G}^H \tilde{\bm{U}}
\end{array}\right] \left[\begin{array}{c}
\bm{J}_0^H \bm{G}^H \tilde{\bm{U}} \\
\vdots \\ \bm{J}_L^H \bm{G}^H \tilde{\bm{U}}
\end{array}\right]^H\\\notag& \cdot\left( \bm{I}_{L+1} \otimes \bm{x}_N\right).
\end{align}
Define the Hankel matrix
\begin{equation}
\tilde{\mathcal{U}}_j \triangleq \left[ \begin{array}{cccc}
u_{0,j} & u_{1,j} & \ldots & u_{L,j} \\
u_{1,j} & u_{2,j} & \ldots & u_{L+1,j} \\
\vdots & \vdots & \ddots & \vdots \\
u_{PN-1,j} & u_{PN, j} & \ldots & u_{PN+L-1,j}
\end{array} \right],
\end{equation}
where $u_{i,j}$ denotes the $(i,j)$th element of the matrix
$\bm{G}^H\tilde{\bm{U}}$. Notice that $u_{i,j}=0$ for $0\leq i\leq
L-1$ and $PN\leq i\leq PN+L-1$ due to the structure of matrix
$\bm{G}$. The matrix $\bm{D}$ can then be written as in
(\ref{eq_simp_D_first}). Noting that $\left( \bm{B}^T \otimes \bm{A}
\right) \mathsf{vec}\left( \bm{X} \right) = \mathsf{vec}\left(
\bm{A}\bm{X}\bm{B} \right)$ for any three matrices $\bm{A}$,
$\bm{B}$, and $\bm{X}$ \cite{Horn1985}, we have (\ref{eq_D_as_vec}).
Since $\bm{x}_N^H \tilde{\mathcal{U}}_i$ is only a row vector for
all $i$, the above equation is equivalent to
\addtocounter{equation}{2}
\begin{align}
\bm{D} &= \frac{1}{\sigma^2} \left[\begin{array}{ccc} \tilde{\mathcal{U}}_0^T
  \bm{x}_N^* & \ldots & \tilde{\mathcal{U}}_{(N-1)L-1}^T  \bm{x}_N^* \end{array}\right] \notag \\
&\quad \left[\begin{array}{ccc} \tilde{\mathcal{U}}_0^T  \bm{x}_N^*
&
\ldots & \tilde{\mathcal{U}}_{(N-1)L-1}^T  \bm{x}_N^* \bm{s}_N^* \end{array}\right]^H \notag \\
&= \frac{1}{\sigma^2} \tilde{\mathcal{U}} \left( \bm{I}_{(N-1)L}
\otimes \bm{x}_N^* \right) \left( \bm{I}_{(N-1)L} \otimes \bm{x}_N^*
\right) \tilde{\mathcal{U}}^H,
\end{align}
where
\begin{equation}
\tilde{\mathcal{U}} \triangleq \left[\begin{array}{cccc}
\tilde{\mathcal{U}}_0^T &  \tilde{\mathcal{U}}_1^T & \ldots &
\tilde{\mathcal{U}}_{(N-1)L-1}^T
\end{array}\right].\label{eq:Utilde}
\end{equation}
Using the fact that $\left( \bm{A} \otimes \bm{B} \right)^T =
\bm{A}^T \otimes \bm{B}^T$ \cite{Horn1985}, we can further simplify
the expression for $\bm{D}$ as
\begin{align}
\bm{D} &= \frac{1}{\sigma^2} \tilde{\mathcal{U}} \left[
\bm{I}_{(N-1)L} \otimes \left( \bm{x}_N^* \bm{x}_N^T  \right)
\right] \tilde{\mathcal{U}}^H. \label{eq:D}
\end{align}

Finally, substituting (\ref{eq:D}) into (\ref{eq:crlb}) and using
(\ref{eq:xN}), we get the final form of the CRB for blind channel
estimation in redundant block transmission systems, as shown in
(\ref{eq_crb_blockTrans}).

\bibliographystyle{IEEEtran}
\bibliography{kkk}

\end{document}